\documentclass[jmp, amsmath,amssymb]{revtex4-1}

\usepackage[T1]{fontenc}
\usepackage[latin9]{inputenc}
\usepackage{geometry}
\usepackage{babel}
\usepackage{amsbsy}
\usepackage{amstext}
\usepackage{amsthm}
\usepackage{amssymb}

\usepackage[table]{xcolor}
\usepackage{xcolor}

\usepackage{graphicx}
\usepackage[colorlinks=true, allcolors=blue,bookmarks=false]{hyperref}

\makeatletter
\numberwithin{equation}{section}
\numberwithin{figure}{section}

\theoremstyle{plain}
\newtheorem{thm}{\protect\theoremname}[section]
\theoremstyle{plain}
\newtheorem{prop}[thm]{\protect\propositionname}
\theoremstyle{remark}

\theoremstyle{definition}
\newtheorem{defn}[thm]{\protect\definitionname}
\theoremstyle{plain}
\newtheorem{lem}[thm]{\protect\lemmaname}

\makeatother

\providecommand{\propositionname}{Proposition}
\providecommand{\remarkname}{Remark}
\providecommand{\theoremname}{Theorem}
\providecommand{\definitionname}{Definition}
\providecommand{\lemmaname}{Lemma}
\begin{document}

\title{Quadratic pseudospectrum for identifying localized states}

\author{Alexander Cerjan}
\email[]{awcerja@sandia.gov}
\affiliation{Center for Integrated Nanotechnologies\\
Sandia National Laboratories\\
Albuquerque, New Mexico 87185, USA}

\author{Terry Loring}
\email{loring@math.unm.edu}
 \affiliation{Department of Mathematics and Statistics\\
University of New Mexico\\
Albuquerque, New Mexico, 87123, USA}

\author{Fredy Vides}
 \email{fredy.vides@unah.edu.hn}
 
\affiliation{ 
Scientific Computing Innovation Center\\
School of Mathematics and Computer Science\\
Universidad Nacional Aut\'onoma de Honduras\\
Tegucigalpa, Honduras
}

\date{\today}
\begin{abstract}
We examine the utility of the quadratic pseudospectrum in photonics and condensed matter. Specifically, the quadratic pseudospectrum represents a method for approaching systems with incompatible observables, as it both minimizes the ``eigen-error'' in the joint approximate spectrum of the incompatible observables and does not increase the system's computational complexity. Moreover, we derive an important estimate relating the Clifford and quadratic pseudospectra. Finally, we prove that the quadratic pseudospectrum is local, and derive the bounds on the errors that are incurred by truncating the system in the vicinity of where the pseudospectrum is being calculated.
\end{abstract}

\maketitle

\section{Introduction}

Many fields in modern physics are faced with the challenge of trying to glean information from incompatible observables in a system. Although this problem is most closely associated with the Heisenberg uncertainty principle in quantum mechanics, it also commonly manifests in classical systems governed by a wave equation. As an example, consider a point defect in a crystalline lattice that can host localized states \cite{harrison_solid_1980}. For such a system, the most important questions are: (1) what are the energies of the defect states, and (2) what is their spatial extent. Unfortunately, these questions correspond to incompatible observables, as in general the Hamiltonian for a crystal, $H$, does not commute with the crystal's position operators, $P_i$ with $i = x,y,z$, $[H, P_i] \ne 0$. Typically, this problem is approached by first finding the spectrum of $H$ over some suitably large volume containing a single defect to find the energies of the defect's states. If the energy of a defect state, $|\psi_{\textrm{def}} \rangle$, is in a bulk band gap of the surrounding crystal, a measure of the location and localization of this state can then be determined using moments of the state's position expectation values, $\langle \psi_{\textrm{def}}| P_i^n |\psi_{\textrm{def}} \rangle$. However, when a defect state's energy is within the extent of the crystal's bulk bands, this approach is no longer possible, as any state associated with the defect is now a member of a large degenerate subspace without a discriminant for choosing a preferred basis. (Different choices of basis will yield different position expectation values.)

One approach for finding an approximate joint eigenspectrum between non-commuting operators is to study the system using pseudospectral methods that do not require directly measuring any of the system's incompatible observables individually. Intuitively, this approach is based on constructing a single composite operator out of the various eigenvalue problems, $(X_j- \lambda_j)\boldsymbol{v}$, for each of the relevant non-commuting operators, $X_j$, and then analyzing the spectrum (or related aspects) of this composite operator. One example of such a composite operator is the \textit{localizer} \cite{LoringPseudospectra},
\begin{equation}
    L_{\boldsymbol{\lambda}}(X_1,\cdots,X_{d})
    =\sum_{j=1}^{d} (X_{j}-\lambda_{j})\otimes\Gamma_{j}, \label{eq:locDef}
\end{equation}
which combines the underlying eigenvalue equations using a non-trivial Clifford representation, $\Gamma_j^\dagger = \Gamma_j$, $\Gamma_j^2 = I$, and $\Gamma_j \Gamma_l = -\Gamma_l \Gamma_j$ for $j \ne l$. Here, we assume that $X_j$ are Hermitian, so that $\lambda_j \in \mathbb{R}$ and thus
\begin{equation*}
    \boldsymbol{\lambda} = (\lambda_1,\cdots,\lambda_d) \in \mathbb{R}^{d}.
\end{equation*}
We want to know if $L_{\boldsymbol{\lambda}}(X_1,\cdots,X_{d})$ is singular, and it not, how far it deviates from singular.
Note that there has been a subtle shift in the treatment of $\lambda_j$ in Eq.\ (\ref{eq:locDef}) as compared to its use in an eigenvalue equation. Whereas in an eigenvalue problem, $(X_j- \lambda_j)\boldsymbol{v} = 0$, $\lambda_j$ is something that is calculated using a known operator (i.e., $\lambda_j$ is a dependent variable), in a composite operator $\boldsymbol{\lambda}$ is better thought of as an input (i.e., as a set of independent variables). This is because, no matter how close the localizer is to being singular, its spectrum contains valuable information on the durability of edge modes \cite{michalaLorWat2020wavePropagation} or on the number of Dirac or Weyl points \cite{schulz-baldes_spectral_2022}.

Even when there is no topology to study, the utility of composite operators is providing for incompatible observables some metric for the inherent uncertainty in joint measurement given a state whose expection in these observables is close to $\boldsymbol{\lambda}$. For example, at a given $\boldsymbol{\lambda}$ the localizer can be used to define the \textit{localizer gap} of $(X_1,\cdots,X_{d})$ as
\begin{equation*}
    \mu_{\boldsymbol{\lambda}}^{C}(X_{1},\cdots,X_{d})=\sigma_{\min}\left(L_{\boldsymbol{\lambda}}(X_1,\cdots,X_{d}) \right).
\end{equation*}
Here we use $\sigma_{\min}(A)$ to denote the smallest singular value of $A$. (Mainly we use this in the case where $A$ is Hermitian and so $\sigma_{\min}(A)$ is the smallest absolute value of an eigenvalue.) Finally, the localizer can be used to define the \textit{Clifford $\epsilon$-pseudospectrum} \cite[\S 1]{LoringPseudospectra}, 
\begin{equation*}
    \Lambda_{\epsilon}^{C}(X_1,\cdots,X_{d}) = \{ \boldsymbol{\lambda} \; | \; \mu_{\boldsymbol{\lambda}}^{C}(X_{1},\cdots,X_{d}) \le \epsilon \} 
\end{equation*}
which is a closed subset of $\mathbb{R}^{d}$.  When $\epsilon=0$ this set is known simply as the Clifford spectrum.  
Those $\boldsymbol{\lambda}$ that yield small localizer gaps correspond to joint approximate eigenvalues of $(X_1,\cdots,X_{d})$. Thus, the utility of composite operators and pseudospectral methods can be understood as enabling the simultaneous, but approximate, joint measurement of many incompatible observables. Note though, that even if  $\boldsymbol{\lambda}$ is a member of the Clifford spectrum of $(X_1,\cdots,X_{d})$, it does not follow that $\lambda_j$ is necessarily an eigenvalue of $X_j$.

\begin{figure}[t]
    \centering
   \includegraphics{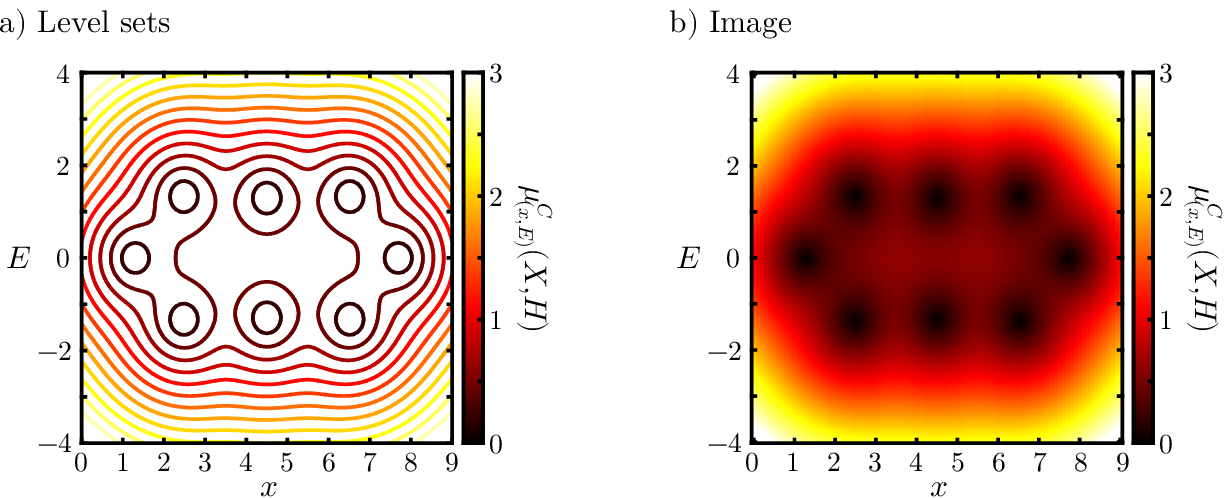}
    \caption{The Clifford pseudospectrum associated to the two observables defining  a basic SSH model.  This is shown as the traditional level sets and as an image with artificial color.  The value of the indicator function hits zero at eight points, two of which correspond to defect states at zero energy.
    \label{fig:Basic_SSH_PS}}
\end{figure} 

Traditionally, pseudospectra are displayed by curves indicating the boundaries of several different $\epsilon$-pseudospectra. These are the level curves of the function
\begin{equation*}
    \boldsymbol{\lambda} \mapsto \mu_{\boldsymbol{\lambda}}^{C}(X_{1},\cdots,X_{d})
\end{equation*}
We refer to this as the \textit{indicator function} for the Clifford pseudospectrum.  We prefer to display the indicator function as an image 
as this is closer to how most synthetic and experimental data are presented in
physics.  In Figure~\ref{fig:Basic_SSH_PS} we illustrate the 
two-variable Clifford pseudospectrum for $(X,H)$, the Hamiltonian and  position
observable for a standard finite Su-Schrieffer-Heeger (SSH) model \cite{su1979solitons} (see, for example, \cite[\S 1.1]{asboth2016short_course_top_ins}).
\begin{equation}
\label{eqn:Basic_SSH_matrices}
  H=\left[\begin{array}{cccccccc}
0 & v\\
v & 0 & w\\
 & w & 0 & v\\
 &  & v & 0 & w\\
 &  &  & w & 0 & v\\
 &  &  &  & v & 0 & w\\
 &  &  &  &  & w & 0 & v\\
 &  &  &  &  &  & v & 0
\end{array}\right],\quad X=\left[\begin{array}{cccccccc}
1\\
 & 2\\
 &  & 3\\
 &  &  & 4\\
 &  &  &  & 5\\
 &  &  &  &  & 6\\
 &  &  &  &  &  & 7\\
 &  &  &  &  &  &  & 8
\end{array}\right]
\end{equation}
with $ v=0.7$ and $w=1.4$.  This will have a defect state at each end reflecting its non-trivial $K$-theory. 

Previously, the Clifford pseudospectrum has attracted interest due to its close connection with a system's $K$-theory \cite{LoringSchuBa_odd,LoringSchuBa_even}, which can be used to diagnose a material's topological properties \cite{xiao_berry_2010,hasan_colloquium:_2010,chiu_classification_2016,bansil_colloquium_2016,ozawa_topological_2019}. When $X_j$ are $H$ and $P_i$ for a crystal, the localizer can be used to determine the crystal's topology in all (physical) dimensions and every symmetry class, regardless of whether the system exhibits a bulk bandgap \cite{cerjan_local_2021}. Moreover, changes in the system's topology can only occur when the localizer gap closes, i.e., at $\boldsymbol{\lambda}$ where $\mu_{\boldsymbol{\lambda}}^{C}(X_{1},\cdots,X_{d})=0$. However, despite these beneficial properties of the localizer, its structure necessarily increases the computational complexity of solving for the system's approximate spectrum beyond the complexity of finding the Hamiltonian's exact spectrum due to the need to tensor the constituent operators with a Clifford representation. (This can be especially problematic in 3D systems where an $n$-by-$n$-by-$n$ system's localizer has size at least $4n^3$-by-$4n^3$.) The increased computational complexity of the localizer presents two related questions: (1) Given that a wide range of physical systems do not possess non-trivial topology, is there a composite operator which can be used to efficiently solve for approximate joint spectra? (2) Even for systems that exhibit topological behaviors where the localizer might be necessary, can we estimate the size of the localizer gap using the pseudospectrum of a composite operator with the same size as the system's Hamiltonian?

Here, we establish the physical relevance of the \textit{quadratic composite operator},
\begin{equation*}
Q_{\boldsymbol{\lambda}}(X_1,\cdots,X_{d})
    =\sum_{j=1}^{d} (X_{j}-\lambda_{j})^2,
\end{equation*}
which can be used to calculate what is called the quadratic pseudospectrum \cite{LoringVides2022PJADE}. The quadratic pseudospectrum is a collection of sets determined by an indicator function, defined below. The value of this indicator function at $\boldsymbol{\lambda}$ tells us about how small the variances in all the observables can be for a state centered at $\boldsymbol{\lambda}$. By centered, we mean $|\psi \rangle$ has expectation value $\lambda_j$ in the $j$th observable. The value of the indicator function can be computed in terms of an eigenvalue of $Q_{\boldsymbol{\lambda}}$. Moreover, we provide a bound on the maximum difference between the smallest singular values of the quadratic composite operator and the localizer. Finally, we prove a number of results that shows the quadratic pseudospectrum is well-behaved in a number of ways that will enable the development of efficient algorithms for its calculation.

The remainder of this paper is organized as follows. In Sec.\ \ref{sec:qps} we provide some basic definitions and results for the quadratic pseudospectrum. In particular, we establish the upper bound between the quadratic and Clifford pseudospectra, proving that at $\boldsymbol{\lambda}$ with sufficiently large quadratic gaps, the localizer gap cannot be zero.  In Sec.\ \ref{sec:Ktheory} we explain why we cannot ignore the Clifford pseudospectrum, namely that the localizer contains $K$-theory while the quadratic composite operator does not. In Sec.\ \ref{sec:examples} we provide some simple examples using small matrices of the behavior of the quadratic and Clifford pseudospectra. In Sec.\ \ref{sec:bulkEdge} we provide some more physically motivated examples of the behavior of both pseudospectra in physical systems with non-trivial topology. In Sec.\ \ref{sec:locality} we prove results regarding the error induced by truncating the spacial extent of the system. Finally, in Sec.\ \ref{sec:conc} we offer some concluding remarks.

\section{Properties of the quadratic pseudospectrum \label{sec:qps}}

We recall that a unit vector $\boldsymbol{v}$ in Hilbert space determines a probability distribution with respect to a Hermitian matrix (observable) $X$.  We do not need access to the full distribution, but only its expectation and variance.  This will be enough for us to talk in fuzzy terms about the location of a state or its approximate energy.  In mathematical notation,
the expectation is
\begin{equation*}
    \textrm{E}_{\boldsymbol{v}}[X] = \langle X \boldsymbol{v}, \boldsymbol{v} \rangle .
\end{equation*}
In physics notation the expectation would be written
\begin{equation*}
    \langle X \rangle_\psi = \langle \psi | X | \psi \rangle .
\end{equation*}
The square of the variance is
\begin{equation*}
    \Delta^2_{\boldsymbol{v}}X = \langle X^2 \boldsymbol{v}, \boldsymbol{v} \rangle - \langle X \boldsymbol{v}, \boldsymbol{v} \rangle^2 .
\end{equation*}
In physics notation, that would look like
\begin{equation*}
    \Delta^2_\psi X =  \langle \psi | X^2 | \psi \rangle - \langle \psi | X | \psi \rangle^2.
\end{equation*}

If we cannot find an eigenvector and eigenvalue, we can try instead to make 
$ \left\Vert A\boldsymbol{v}-\lambda\boldsymbol{v}\right\Vert  $
as small as possible.  This ``eigen-error'' occurs naturally in mathematics so long as $\boldsymbol{v}$ is a unit vector. Assuming as well that $\lambda$ is real, we compute 
\begin{equation*}
\left\Vert X\boldsymbol{v}-\lambda\boldsymbol{v}\right\Vert ^{2}  = \left\langle X^{2}\boldsymbol{v},\boldsymbol{v}\right\rangle -2\lambda\left\langle X\boldsymbol{v},\boldsymbol{v}\right\rangle +\lambda^{2} 
  =\left(\left\langle X^{2}\boldsymbol{v},\boldsymbol{v}\right\rangle -\left\langle X\boldsymbol{v},\boldsymbol{v}\right\rangle ^{2}\right)+\left(\left\langle X\boldsymbol{v},\boldsymbol{v}\right\rangle -\lambda\right)^{2}
\end{equation*}
we are able to rewrite this expression as
\begin{equation}
\label{eqn:eigenerror-rewritten}
\left\Vert X\boldsymbol{v}-\lambda\boldsymbol{v}\right\Vert ^{2} =
\Delta_{\boldsymbol{v}}^{2}X + \left(\textrm{E}_{\boldsymbol{v}}[X] - \lambda\right)^{2}
\end{equation}
which is the square sum of the variance and the
displacement of the expectation from what we were expecting.

As in general $[X_j,X_l] \ne 0$, it is impossible to find exact joint eigenvectors of these operators. The quadratic pseudospectrum instead provides a measurement of how close we can get to a joint eigenvector. The following is built on the ideas in  \cite[\S~6]{schneiderbauer2016Quasicoherent_states}.  This theorem relates physically relevant features of states to matrix computations that allow for reasonably fast numerical algorithms.

\begin{prop} \label{prop:quad_gap_four_ways}
Suppose $X_{1},\dots,X_{d}$ are $n$-by-$n$ Hermitian matrices and that $\boldsymbol{\lambda}$ is an element of $\mathbb{R}^{d}$.
The following quantities are always equal:
\begin{enumerate}
\item the minimum of
\begin{equation*}
\sqrt{\sum_{j=1}^{d}\left\Vert X_{j}\boldsymbol{v}-\lambda_{j}\boldsymbol{v}\right\Vert ^{2}}
\end{equation*}
as $\boldsymbol{v}$ ranges over all unit vectors in $\mathbb{C}^{n}$,
\item the minimum of
\begin{equation*}
\sqrt{\sum_{j=1}^{d}\Delta_{\boldsymbol{v}}^{2}X_{j}+\sum_{j=1}^{d}\left(\textrm{E}_{\boldsymbol{v}}[X_j] - \lambda\right)^{2}}
\end{equation*}
as $\boldsymbol{v}$ ranges over all unit vectors in $\mathbb{C}^{n}$,
\item the smallest singular value of 
\begin{equation}
\label{eqn:Tall_skinny_composite}
M_{\boldsymbol{\lambda}}(X_1,\cdots,X_{d}) =
\left[\begin{array}{c}
X_{1}-\lambda_{1}\\
X_{2}-\lambda_{2}\\
\vdots\\
X_{d}-\lambda_{d}
\end{array}\right],
\end{equation}
\item the square root of the smallest absolute value of an eigenvalue of $Q_{\boldsymbol{\lambda}}(X_1,\cdots,X_{d})$.
\end{enumerate}
Moreover, a unit vector is a right singular vector of \textup{(\ref{eqn:Tall_skinny_composite})} iff it is an eigenvector of $Q_{\boldsymbol{\lambda}}(X_1,\cdots,X_{d})$, iff it minimizes the quantity in \textup{(1)}, iff it minimizes the quantity in \textup{(2)}.
\end{prop}

\begin{proof}
The equality of (1) and (2) follows from (\ref{eqn:eigenerror-rewritten}).
The equality of (3) and (4) follows from
\begin{equation*}
    (M_{\boldsymbol{\lambda}}(X_1,\cdots,X_{d}))^\dagger M_{\boldsymbol{\lambda}}(X_1,\cdots,X_{d}) = Q_{\boldsymbol{\lambda}}(X_1,\cdots,X_{d})
\end{equation*}
and the fact that
$   \sigma_{\min}(A^\dagger A) = (\sigma_{\min}(A))^2$
for any matrix $A$.
The final part of the argument uses a characterisation of the smallest singular value of a matrix $A$  (see \cite[Thm.~8.6.1]{Golub_VanLoan_MatrixComp})
as
\begin{equation*}
    \sigma_{\min}(A) = \min_{\|\boldsymbol{v}\|=1} \|A\boldsymbol{v}\|
\end{equation*}
and the routine calculation
\begin{equation*}
    \| M_{\boldsymbol{\lambda}}(X_1,\cdots,X_{d}) \boldsymbol{v} \| = \sqrt{\sum \|X_j \boldsymbol{v} - \lambda_j \boldsymbol{v} \|^2}
\end{equation*}
\end{proof}

\begin{defn}
Suppose $X_{1},\dots,X_{d}$ are Hermitian matrices in $\boldsymbol{M}_{n}(\mathbb{C})$.
For every 
$\boldsymbol{\lambda}=\left(\lambda_{1},\dots,\lambda_{d}\right)$ in $\mathbb{R}^{d}$
we define the \emph{quadratic gap }of $(X_{1},\dots,X_{d})$ at
$\boldsymbol{\lambda}$ as
\begin{equation*}
    \mu_{\boldsymbol{\lambda}}^{Q}(X_{1},\dots,X_{d}) = \left(\sigma_{\min}(Q_{\boldsymbol{\lambda}}(X_1,\cdots,X_{d}))\right)^{\frac{1}{2}}.
\end{equation*}
The \textit{quadratic $\epsilon$-pseudospectrum} of $(X_{1},\dots,X_{d})$
is defined as the set
\begin{equation*}
    \Lambda_{\varepsilon}^{Q}(X_1,\cdots,X_{d}) = \{ \boldsymbol{\lambda} \; | \; \mu_{\boldsymbol{\lambda}}^{Q}(X_{1},\cdots,X_{d}) \le \varepsilon \} .
\end{equation*}
  When $\epsilon=0$ this set is known simply as the quadratic spectrum.  
We call the function 
\begin{equation*}
\boldsymbol{\lambda}\mapsto\mu_{\boldsymbol{\lambda}}^{Q}(X_{1},\dots,X_{d})
\end{equation*}
 the \textit{indicator function} of the quadratic pseudospectrum.
\end{defn}

\begin{figure}[t]
    \centering
   \includegraphics{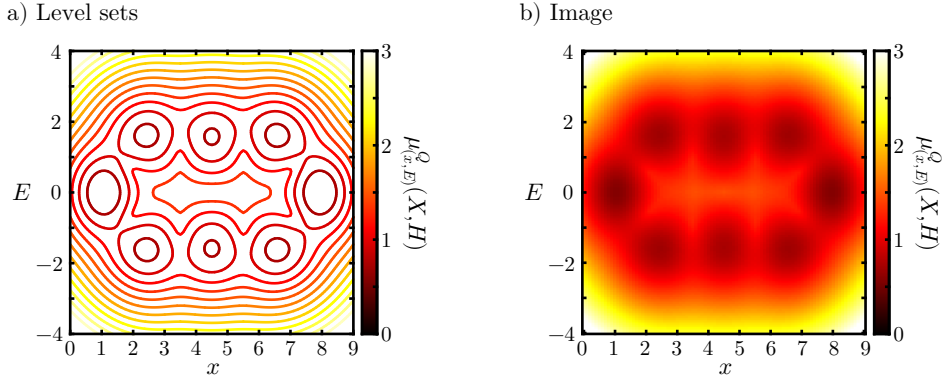}
    \caption{The quadratic pseudospectrum associated to the two observables defining  a basic SSH model.  This is shown as the traditional level sets and as an image with artificial color.  The value of the indicator function is never zero.
    \label{fig:Basic_SSH_Quad_PS}}
\end{figure}

Some basic results about eigenvectors can be tweaked to work with approximate eigenvectors. The following lemma is an example, a modification of the usual fact that for Hermitian matrices, different eigenspaces are orthogonal.  Notice there is no sensible interpretation of a subspace of approximate eigenvectors for a fixed scalar.  We omit the proof as it is an easy modification of a standard short proof.

\begin{lem}
Suppose that $A$ is a Hermitian matrix and that $\boldsymbol{v}$ and $\boldsymbol{w}$
are unit vectors and $\lambda \neq \mu$  are two real numbers.  Then
\begin{equation*}
   \left|\left\langle \boldsymbol{v},\boldsymbol{w}\right\rangle \right| \leq
   \frac{\left\Vert A\boldsymbol{v}-\lambda\boldsymbol{v}\right\Vert +\left\Vert A\boldsymbol{w}-\mu\boldsymbol{w}\right\Vert }{|\lambda-\mu|}. 
\end{equation*}
\end{lem}

A strange thing is that the  quadratic spectrum of $(X_{1},\dots,X_{d})$ is often
the empty set.  One example of this is worked out in 
\cite{DebLorSverSurfaces}.  We find  another example now by working out the quadratic pseudospectrum for the matrices
from (\ref{eqn:Basic_SSH_matrices}), as illustrated in Figure~\ref{fig:Basic_SSH_Quad_PS}.

However, we are most interested in the places where the quadratic
spectrum is close to or at a local minimum value even if that minimum is not zero.

There are instances where one needs to compute both the Clifford and the quadratic pseudospectra.  In general, the
Clifford pseudospectrum is more interesting at $\boldsymbol{\lambda}$ where $\mu_{\boldsymbol{\lambda}}^{C}$ is large.  One can use the faster quadratic pseudospectrum
to determine where to avoid computing the Clifford pseudospectrum whenever these two are close.  The following
is thus a useful bound.

\begin{prop} \label{prop:Diff_PS_bound}
If $X_{1},\dots,X_{d}$ are Hermitian matrices in $\boldsymbol{M}_{n}(\mathbb{C})$
and 
$\boldsymbol{\lambda} \in\mathbb{R}^{d}$
then
\begin{equation}
\left|\left(\mu_{\boldsymbol{\lambda}}^{Q}(X_{1},\dots,X_{d})\right)^2-\left(\mu_{\boldsymbol{\lambda}}^{C}(X_{1},\dots,X_{d})\right)^2\right|
\leq
\sum_{j<k}\left\Vert \left[X_{j},X_{k}\right]\right\Vert .
\label{PS_to_PS}
\end{equation}
\end{prop}

\begin{proof}
Since tensoring by the identity will not alter the spectrum of a matrix,
\begin{equation*}
\sigma_{\min}\left(\sum(X_{j}-\lambda_{j})^{2}\right)=\sigma_{\min}\left(\sum(X_{j}-\lambda_{j})^{2}\otimes I\right).
\end{equation*}
We also have the estimate 
\begin{equation}
\left(L_{\boldsymbol{\lambda}}(X_{1},\cdots,X_{d})\right)^{2}=\sum(X_{j}-\lambda_{j})^{2}\otimes I+\sum_{j<k}\left[X_{j},X_{k}\right]\otimes\Gamma_{j}\Gamma_{k}
\label{difference_quad_Cliff}
\end{equation}
and so 
\[
\left\Vert \left(L_{\boldsymbol{\lambda}}(X_{1},\cdots,X_{d})\right)^{2}-\sum(X_{j}-\lambda_{j})^{2}\otimes I\right\Vert \leq\sum_{j<k}\left\Vert \left[X_{j},X_{k}\right]\right\Vert .
\]
As this is an operator norm estimate between two Hermitian matrices,
we find that 
\begin{equation*}
\left|\sigma_{\min}\left(L_{\boldsymbol{\lambda}}(X_{1},\cdots,X_{d})\right)^{2}-\sigma_{\min}\left(\sum(X_{j}-\lambda_{j})^{2}\otimes I\right)\right|\leq\sum_{j<k}\left\Vert \left[X_{j},X_{k}\right]\right\Vert .
\end{equation*}
\end{proof}

A difficultly in determining either pseudospectra of a noncommutative $d$-tuple is that calculating the value at one or more values of $\boldsymbol{\lambda}$ does not provide much assistance in calculating the value at another value.  We do, at least, have a sense of how fast
$\mu^C_{\boldsymbol{\lambda}}$ or $\mu^Q_{\boldsymbol{\lambda}}$ can vary, so can limit the number of vales of $\boldsymbol{\lambda}$ that need to be considered. We know from results in \cite[\S 7]{LoringPseudospectra} that $\mu^C_{\boldsymbol{\lambda}}$ is Lipschitz with constant $1$.  We next show that the same is true in the quadratic case.

\begin{prop}
Suppose $X_{1},\dots,X_{d}$ are Hermitian matrices in $\boldsymbol{M}_{n}(\mathbb{C})$.  If
$\boldsymbol{\lambda},\boldsymbol{\mu} $ are two elements of $\mathbb{R}^{d}$
\begin{equation*}
    \left| \mu_{\boldsymbol{\lambda}}^{Q}(X_{1},\dots,X_{d})-\mu_{\boldsymbol{\nu}}^{Q}(X_{1},\dots,X_{d}) \right|
    \leq
    \| \boldsymbol{\lambda} - \boldsymbol{\nu} \|
\end{equation*}
were the norm on the right is the Euclidean norm.
\end{prop}

\begin{proof}
It is well known that all singular values are Lipshitz in the matrix input $A$, so long as one uses the operator norm to give the metric on the space of $m$-by-$n$ matrices.  For example, one can apply Weyl's inequality to the eigenvalues of
the Hermitian matrix 
\begin{equation*}
    \left[\begin{array}{cc}
0 & A\\
A^{\dagger} & 0
\end{array}\right].
\end{equation*}
Finally, 
\begin{align*}
\left\Vert M_{\boldsymbol{\lambda}}(X_{1},\cdots,X_{d})-M_{\boldsymbol{\mu}}(X_{1},\cdots,X_{d})\right\Vert ^{2} & =\left\Vert \left[\begin{array}{c}
(\lambda_{1}-\mu_{1})I\\
\vdots\\
(\lambda_{d}-\mu_{d})I
\end{array}\right]\right\Vert ^{2}\\
 & =\left\Vert \left[\begin{array}{ccc}
(\lambda_{1}-\mu_{1})I & \cdots & (\lambda_{d}-\mu_{d})I\end{array}\right]\left[\begin{array}{c}
(\lambda_{1}-\mu_{1})I\\
\vdots\\
(\lambda_{d}-\mu_{d})I
\end{array}\right]\right\Vert \\
 & =\left\Vert \left(\sum(\lambda_{1}-\mu_{1})^{2}\right)I\right\Vert \\
 & = \sum(\lambda_{1}-\mu_{1})^{2}.
\end{align*}
\end{proof}

\section{Where is the K-theory? \label{sec:Ktheory}}

For all its advantages, the quadratic pseudospectrum seems to not see $K$-theory.  To get a sense of why, we consider the straight path of
systems with $H_t$ as in 
(\ref{eqn:Basic_SSH_matrices})
except now with
\begin{equation*}
    v = 0.7(1-t) + 1.4t
\end{equation*}
and
\begin{equation*}
    w = 1.4(1-t) + 0.7t .
\end{equation*}

We fix $\boldsymbol{\lambda} = (4,0) $ so we are looking at zero energy and at the approximate middle in position.  (Staying way from the exact middle eliminates a mirror symmetry that confuses things.)
Instead of just looking at the smallest singular value of
$
    L_{(4,0)}(X,H_t)
$
and
$
    Q_{(4,0)}(X,H_t)
$
we look at the entire spectrum of this composite operators.  For easier comparison, we actually will plot the square root of the (positive) eigenvalues of the quadratic composite operator.

Since $[H_t,X]$ is not zero, and indeed has no null space, the additive form of the uncertainty principal \cite{maccone2014SumOfUncertainty}, combined with Proposition~\ref{prop:quad_gap_four_ways}, will tell us that
quadratic composite operator is never singular.  Thus we are looking at a path of invertible matrices and we cannot use it to detect a change of any topological index. See Figure~\ref{fig:SSH_path}-a.

The spectrum of the localizer, shown in Figure~\ref{fig:SSH_path}-b, is more promising.  However, the upward and downward moving eigenvalues appear as if they should somehow cancel out.  We explain below that there is a symmetry here in the localizer when $d=2$ which is forcing the entire spectrum of the localizer to be symmetric about zero.

Let 
\begin{equation*}
\Gamma=\left[\begin{array}{cccccccc}
1\\
 & -1\\
 &  & 1\\
 &  &  & -1\\
 &  &  &  & 1\\
 &  &  &  &  & -1\\
 &  &  &  &  &  & 1\\
 &  &  &  &  &  &  & -1
\end{array}\right]
\end{equation*}
be the grading operator for which we have $H_{t}\Gamma=-\Gamma H_{t}$ and $X\Gamma=\Gamma X$, reflecting the chiral nature of this system. 
Let $\tilde{X}=X-4I$. We find
\begin{equation*}
L_{(0,4)}(X,H_{t})=\left[\begin{array}{cc}
0 & \tilde{X}-iH_{t}\\
\tilde{X}+iH_{t} & 0
\end{array}\right]=\left[\begin{array}{cc}
\Gamma\\
 & I
\end{array}\right]\left[\begin{array}{cc}
0 & \left(\left(\tilde{X}+iH_{t}\right)\Gamma\right)^{*}\\
\left(\tilde{X}+iH_{t}\right)\Gamma & 0
\end{array}\right]\left[\begin{array}{cc}
\Gamma\\
 & I
\end{array}\right]^{*}
\end{equation*}
which tells us that
\begin{equation*}
\alpha\in\sigma\left(L_{(0,4)}(X,H_{t})\right)\iff\pm\alpha\in\sigma\left(\left(\tilde{X}+iH_{t}\right)\Gamma\right).
\end{equation*}
Notice we are using $\sigma_x$ and $\sigma_y$ as the default choice of $\Gamma $ matrices, for reasons we explain in Sec.\ \ref{sec:examples}.

\begin{figure}[t]
    \centering
   \includegraphics{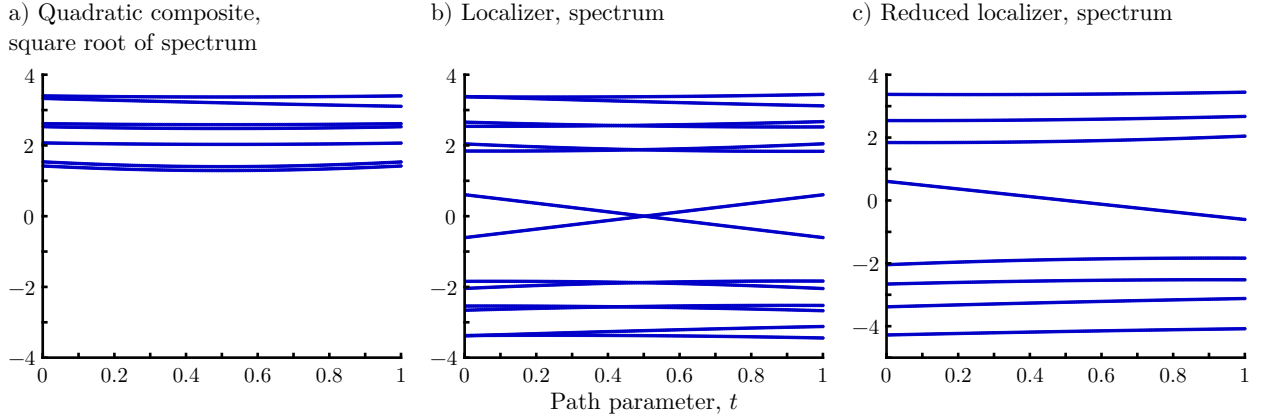}
    \caption{We work along a path from a trivial system at $t=0$ to a topological system at $t=1$.  In (a) we look at the full spectrum of the quadratic composite operator,  fixing $\boldsymbol{\lambda} = (4,0)$.  In (b) we show the spectrum of the localizer.  In (c) we show the spectrum of the reduced localizer, $((X- 4)+iH_{t})\Gamma$.
    \label{fig:SSH_path}}
\end{figure} 

We plot the spectrum of this ``reduced localizer'' in Figure~\ref{fig:SSH_path}-c.  The details of
how this spectral imbalance is a proper $K$-theory invariant can be found in \cite{LoringSchuBa_odd}.  We mention it here to explain that the positivity of the quadratic composite operator makes it immune to $K$-theory.  We say positive as a shorter version of positive semi-definite.

In other symmetry classes, the correct index is found using the sign of the determinant of a matrix that has real elements.  Notice here that $(\tilde{X}+iH_{t})\Gamma$ has real determinant and an odd number of eigenvalue passing from positive to negative would have been reflected in the sign of the determinant.  We will look at another small example in the next section that corresponds to a symmetry class where the expected index is $\mathbb{Z}/2$  index that can be found using the sign of a determinant of a similar matrix.  The details of why this index is a reasonable index are more complicated \cite{DollShuba_skew-localizer}.  Again, the fact that the quadratic composite operator is positive means its determinant is stuck being positive, so cannot detect this sort of invariant either.

\section{Mathematical Examples \label{sec:examples}}

In mathematics, the traditional form of the pseudospectrum arose as a way to
investigate a single matrix $A$ that is not normal \cite{TrefethenEmbree}.
We can write $A$ in the usual way, in terms of two Hermitian matrices
\begin{equation*}
    A = X + iY
\end{equation*}
where $X = \tfrac{1}{2}( A^* + A )$  and $B = \tfrac{i}{2} (A - A^*)$.  The quadratic pseudospectrum 
\begin{equation*}
 \mu_{(x,y)}^{Q}(X,Y) = 
 \min_{\|\boldsymbol{v}\|=1}\sqrt{\left\Vert X\boldsymbol{v}-x\boldsymbol{v}\right\Vert ^{2} + \left\Vert Y\boldsymbol{v}-y\boldsymbol{v}\right\Vert ^{2}}
 \end{equation*}
looks for good approximate eigenvectors for $X$ and $Y$ at the same time, while
the Clifford pseudospectrum
\begin{equation*}
 \mu_{(x,y)}^{C}(X,Y) = 
 \min_{\|\boldsymbol{v}\|=1}\left\Vert A\boldsymbol{v}-(x+iy)\boldsymbol{v}\right\Vert
 \end{equation*}
looks for exact and approximate eigenvectors of $A$ which might correspond to 
complex eigenvalues.  To see that later claim, notice that if we select $\sigma_x$ and $\sigma_y$ as our $\Gamma$ matrices then 
\begin{align*}
    L_{(x,y)}(X,Y) 
&= \left[\begin{array}{cc}
0 & X-iY-(x-iy)\\
X+iY-(x+iy) & 0
\end{array}\right].
\end{align*}
Here we use $\sigma_x$, $\sigma_y$ and $\sigma_z$ to denote the usual 
2-by-2 Pauli matrices.

The traditional pseudospectrum of a non-normal matrix has been applied in physics to assist with the analysis of 
lossy systems \cite{krejvcivrik2015pseudospectra,Makris2021transient_Expept_Points,Okuma2020Zero_Modes_nonnormality,Sivan2022multipleSSH,komis_robustness_2022}.

\subsection{A 2-by-2 pair}
Here we use the Pauli matrices themselves as our example, so we let
\begin{equation*}
    X=\left[\begin{array}{cc}
0 & 1\\
1 & 0
\end{array}\right]
\end{equation*}
and
\begin{equation*}
    Y=\left[\begin{array}{cc}
0 & -i\\
i & 0
\end{array}\right].
\end{equation*}
In this example, 
\begin{equation*}
    A = X+iY = \left[\begin{array}{cc}
0 & 2\\
0 & 0
\end{array}\right]
\end{equation*}
is nilpotent.  Thus the spectrum of $A$ is just $\{0\}\subseteq\mathbb{C}$ and so the Clifford pseudospectrum of $(X,Y)$ has one zero, at $(0,0)\in \mathbb{R}^{2}$.
This is one of the few cases for which we can easily calculate by hand both the Clifford and quadratic pseudospectra.

We need the singular values of 
\begin{equation*}
    A-\lambda=\left[\begin{array}{cc}
-\lambda & 2\\
0 & -\lambda
\end{array}\right]
\end{equation*}
(with $\overline{\lambda}=x+iy)$ so the square roots of the eigenvalues of 
\begin{equation*}
  \left[\begin{array}{cc}
|\lambda|^{2}+4 & -2\overline{\lambda}\\
-2\lambda & |\lambda|^{2}
\end{array}\right].
\end{equation*}
This has characteristic polynomial (in $\alpha$)
\begin{equation*}
\left|\begin{array}{cc}
|\lambda|^{2}+4-\alpha & -2\overline{\lambda}\\
-2\lambda & |\lambda|^{2}-\alpha
\end{array}\right|=\alpha^{2}+\left(-2|\lambda|^{2}-4\right)\alpha+|\lambda|^{4}
\end{equation*}
so the eigenvalues are
\begin{equation*}
\frac{-\left(-2|\lambda|^{2}-4\right)\pm\sqrt{\left(-2|\lambda|^{2}-4\right)^{2}-4\left(|\lambda|^{4}\right)}}{2}=|\lambda|^{2}+2\pm2\sqrt{|\lambda|^{2}+1} .
\end{equation*}
We want the square root of the smaller, so we have computed
\begin{equation*}
   \mu_{(x,y)}^{C}(X,Y)=\sqrt{x^{2}+y^{2}+2-2\sqrt{x^{2}+y^{2}+1}}.
\end{equation*}

As to the quadratic pseudospectrum of (X,Y), we need the square root of the smallest eigenvalue of
\begin{equation*}
(X-x)^{2}+(Y-y)^{2}=\left[\begin{array}{cc}
|\lambda|^{2}+2 & -2\overline{\lambda}\\
-2\lambda & |\lambda|^{2}+2
\end{array}\right] .
\end{equation*}
This has characteristic polynomial
\begin{equation*}
\left|\begin{array}{cc}
|\lambda|^{2}+2-\alpha & -2\overline{\lambda}\\
-2\lambda & |\lambda|^{2}+2-\alpha
\end{array}\right|=\alpha^{2}+\left(-2|\lambda|^{2}-4\right)\alpha+|\lambda|^{4}+4|\lambda|^{2}+4
\end{equation*}
so the eigenvalues are
\begin{equation*}
\frac{-\left(-2|\lambda|^{2}-4\right)\pm\sqrt{\left(-2|\lambda|^{2}-4\right)^{2}-4\left(|\lambda|^{4}+4|\lambda|^{2}+4\right)}}{2}=|\lambda|^{2}+2\pm2|\lambda|.
\end{equation*}
We want the square root of the smaller, so we have computed
\begin{equation*}
\mu_{(x,y)}^{Q}(X,Y)=\sqrt{x^{2}+y^{2}+2-2\sqrt{x^{2}+y^{2}}}.
\end{equation*}

In this example, the Clifford pseudospectrum has minimum value of $0$, attained at the one point $(x,y)=(0,0)$.  In contrast, the quadratic pseudospectrum has minimum value of $1$ attained on the unit circle.  Notice that at $(0,0)$ the difference between the squares of the two pseudospectra is $2$, equaling the norm of the commutator of $A$ and $B$.  Thus the estimate in Proposition~\ref{prop:Diff_PS_bound} is the best possible, at least for the case of a pair of matrices.

It is interesting to note that the Clifford spectrum of
the three Pauli spin matrices is the unit sphere, while the
Clifford spectrum of any two of them is a singleton.  This is
just one of the nonintuitive features of the Clifford spectrum 
\cite{DebLorSverSurfaces}.

\subsection{A 3-by-3 pair \label{subsec:3x3pair}}

A slightly larger example is the pair
\begin{equation*}
X =\left[\begin{array}{ccc}
0 & 1 & 0\\
1 & 0 & 1\\
0 & 1 & 0
\end{array}\right] 
\end{equation*}
and
\begin{equation*}
Y =\left[\begin{array}{ccc}
0 & 1 & 0\\
1 & 0 & 0\\
0 & 0 & 0
\end{array}\right] 
\end{equation*}
corresponding to the one non-normal matrix
\begin{equation*}
A = X+iY = \left[\begin{array}{ccc}
0 & 1+i & 0\\
1+i & 0 & 1\\
0 & 1 & 0
\end{array}\right] 
\end{equation*}

\begin{figure}[t]
    \centering
    \includegraphics{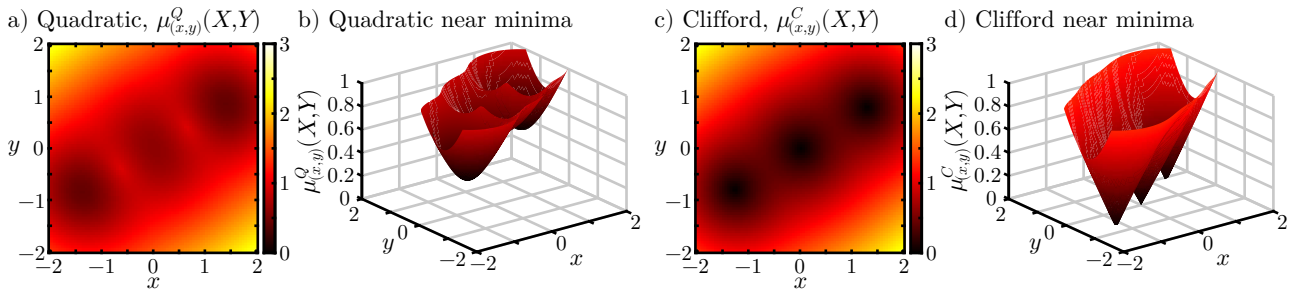}
    \caption{Both the quadratic and Clifford pseudospectra are shown for the 3-by-3 pair of matrices discussed in Sec.\ \ref{subsec:3x3pair}.
    \label{fig:NonHermitianMatrixPS}}
\end{figure}

This has eigenvalues at $0$ and approximately $\pm(1.272+0.786i)$. Figure~\ref{fig:NonHermitianMatrixPS} shows the two pseudospectra for this pair of matrices.

\subsection{A 4-by-4 pair \label{subsec:4x4pair}}

For a second example we let $A=X+iY$ with 
\begin{equation*}
X=\left[\begin{array}{cccc}
-2 & 0 & 0 & 0\\
0 & 2 & 0 & 0\\
0 & 0 & 0 & 0\\
0 & 0 & 0 & 1
\end{array}\right]
\end{equation*}
and
\begin{equation*}
 Y=\left[\begin{array}{cccc}
0 & i & 0 & 0\\
-i & 0 & i & 0\\
0 & -i & 0 & i\\
0 & 0 & -i & 0
\end{array}\right]
\end{equation*}
Since one of our matrices is real and the other purely imaginary, this is reminiscent of
a model of a 1D system in class D \cite{altland_nonstandard_1997,ryu_topological_2010}.  Thus $X+iY$ is real, and the $\mathbb{Z}/2$ topological index related to this example is the sign of the determinant of $X+iY$ \cite{LoringPseudospectra}.  In this case, the determinant is negative, which means we can expect one or three eigenvalues on the negative part of the real axis.  Figure~\ref{fig:ClassDMatrixPS} shows the two pseudospectra for this pair of matrices. The real eigenvalues are approximately $-1.7638$ and $1.4400$ while the complex pair is approximately $0.6619 \pm 1.2371i$.
 
\begin{figure}[t]
    \centering
    \includegraphics{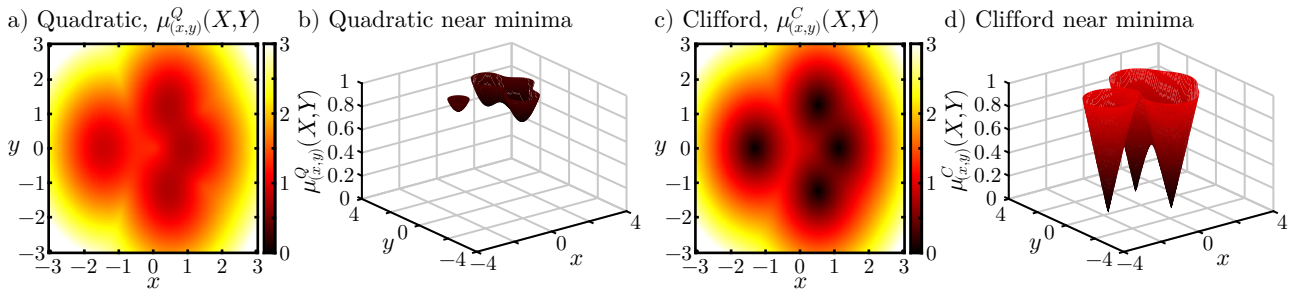}
    \caption{Both the quadratic and Clifford pseudospectra are shown for the 4-by-4 pair of matrices discussed in Sec.\ \ref{subsec:4x4pair} where one of the matrices is real and the other is purely imaginary. 
    \label{fig:ClassDMatrixPS}}
\end{figure}

The quadratic pseudospectrum gives a methodology to attack the problem of finding a unit vector
$\boldsymbol{v}$ in $\mathbb{R}^n$ to minimize
\begin{equation*}
\left\Vert X\boldsymbol{v}-x\boldsymbol{v}\right\Vert ^{2} + \left\Vert Y\boldsymbol{v}-y\boldsymbol{v}\right\Vert ^{2}
 \end{equation*}
given a fixed pair $(x,y)$ of scalars.  Many related questions come to mind,
such as optimizing by varying both $\boldsymbol{v}$ and  $(x,y)$.  One can
also seek joint approximate eigenvectors $\boldsymbol{v}_1,\dots,\boldsymbol{v}_k$ 
for $1< k \leq n$ and perhaps require these vectors to be orthogonal.  The case when $k=n$ is the optimization problem addressed by the JADE algorithm \cite{cardosoSimultanDiagn}.  Moreover, when $k \ll n$ it is possible to combine quadratic pseudospectrum method with JADE to get a small number of orthogonal joint approximate eigenvectors \cite{LoringVides2022PJADE}.  While this may have applications in
physics \cite{RevModPhys.84.1419} we are content to have this paper focus on the problems related to a single 
joint approximate eigenvector.
 
\subsection{AB phase change in Class D}

A slight variation on the last example is more relevant to topological 
insulators.  In a tight-binding model of a 1D system in symmetry class D we
expect a position operator $X$ that is real symmetric (even diagonal) and
a Hamiltonian $H$ that is purely imaginary.  

In this example we work with 
\begin{equation*}
    X=\left[\begin{array}{ccccccc}
\!\!\!-3.5\\
 & \!\!\!\!-2.33\\
 &  & \!\!\!\!-1.17\\
 &  &  & 0\\
 &  &  &  & 1.17\\
 &  &  &  &  & 2.33\\
 &  &  &  &  &  & 3.5
\end{array}\right]
\end{equation*}
 and
\begin{equation*}
    H=\left[\begin{array}{ccccccc}
0 & 1.4i\\
\!-1.4i & 0 & 0.7i\\
 & \!-0.7i & 0 & 0.7i\\
 &  & \!-0.7i & 0 & 1.4i\\
 &  &  & \!-1.4i & 0 & 0.7i\\
 &  &  &  & \!-0.7i & 0 & 1.4i\\
 &  &  &  &  & \!-1.4i & 0
\end{array}\right]
\end{equation*}
This is reminiscent of an AB phase change in the SSH model.  Notice this system has an odd number of sites, in line with previous theory
\cite{su1979solitons} and experiment \cite{meier_SSH_observation,weimann_topologically_2017}.  

The real matrix $X+iH$ has a negative determinant, so it is not
surprising that it has a single real eigenvalue and six conjugate pairs of complex eigenvalues.  What is remarkable is that the negative eigenvalue
is at the location of the AB phase transition, and that it is not easily moved by perturbation.

\begin{figure}
    \centering
    \includegraphics{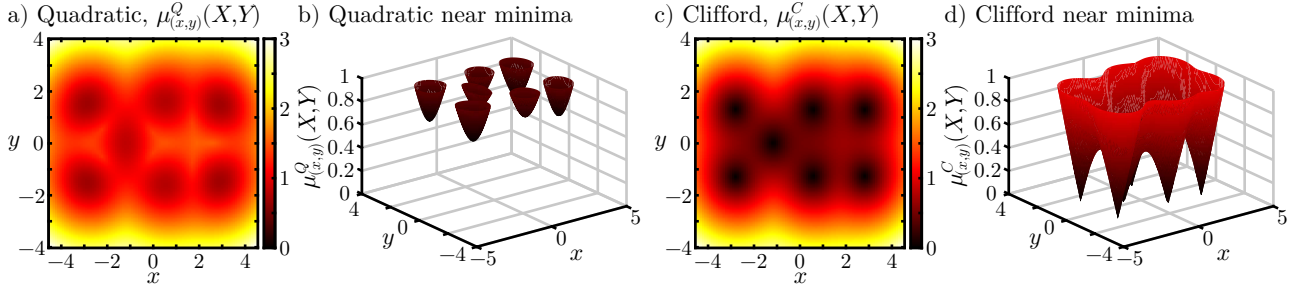}
    \caption{Both the Clifford and quadratic pseudospectra are shown for two small Hermitian that represent a small system in class D with a defect induced by a phase change.
    \label{fig:ClassDMatrix_SSH_likePS}}
\end{figure} 

For the record:  The real eigenvalue is approximately $-1.1603$ while the complex pairs are approximately $-2.7876 \pm 1.2941i$ and $2.8020 \pm 1.2703i$.
 
\subsection{Symmetries in the pseudospectra}
 
 Many of the previous examples had apparent symmetries in both the Clifford and Quadratic pseudospectra, as symmetries in the given matrices tend to lead to symmetries in the pseudospectra. 
 
 It is hard not to see horizontal and vertical symmetry in all the images in Figures~\ref{fig:Basic_SSH_PS} and \ref{fig:Basic_SSH_Quad_PS}. These symmetries are a manifestation of the two symmetries in the underlying SSH systems, mirror symmetry and sublattice (chiral) symmetry. That this leads to symmetry with both forms of pseudospectrum is then a consequence of the following theorem.

\begin{thm}
Suppose $X_{1},\dots,X_{d}$ are Hermitian matrices. Suppose $S$
is a unitary matrix such that $SX_{j}=X_{j}S$ for all $j$ except
that $SX_{d}=-X_{d}S$. Then
\begin{equation*}
\mu_{\boldsymbol{\lambda}}^{C}(X_{1},\dots,X_{d})=\mu_{\boldsymbol{\gamma}}^{C}(X_{1},\dots,X_{d})
\end{equation*}
and 
\begin{equation*}
\mu_{\boldsymbol{\lambda}}^{Q}(X_{1},\dots,X_{d})=\mu_{\boldsymbol{\gamma}}^{Q}(X_{1},\dots,X_{d})
\end{equation*}
 for $\boldsymbol{\gamma}=(\lambda_{1},\lambda_{2},\dots,\lambda_{d-1},-\lambda_{d})$.
\end{thm}
 
 \begin{proof}
Since
\begin{equation*}
S(X_{d}-\lambda_{d})^{2}S^{\dagger}=(-X_{d}-\lambda_{d})^{2}=(X_{d}+\lambda_{d})^{2}
\end{equation*}
and 
\begin{equation*}
S(X_{j}-\lambda_{j})^{2}S^{\dagger}=(X_{j}-\lambda_{j})^{2}
\end{equation*}
for the other $j$ we find that 
\begin{equation*}
SQ_{\boldsymbol{\lambda}}(X_{1},\dots,X_{d})S^{\dagger}=Q_{\boldsymbol{\gamma}}(X_{1},\dots,X_{d}).
\end{equation*}
 These two composite matrices are unitarily equivalent, so they have
the same eigenvalues.

For the localizer, we first fix some choice of the $\Gamma_{j}$.  Notice that
$\Gamma_{1},\Gamma_{2},\dots,\Gamma_{d-1},\Gamma_{d}$ is also
a representation of the generators the appropriate Clifford algebra.  When $d$ is even
there is a unitary matrix $R$ so that $R\Gamma_{d}R^{\dagger}=-\Gamma_{d}$
and $R\Gamma_{j}R^{\dagger}=\Gamma_{j}$ for all other $j$. We now find
\begin{equation*}
\left(S\otimes R\right)\left((X_{d}-\lambda_{d})\otimes\Gamma_{d}\right)\left(S\otimes R\right)^{\dagger}=(-X_{d}-\lambda_{d})\otimes\left(-\Gamma_{d}\right)=(X_{d}+\lambda_{d})\otimes\Gamma_{d}
\end{equation*}
 and, for the other $j$,
\begin{equation*}
\left(S\otimes R\right)\left((X_{j}-\lambda_{j})\otimes\Gamma_{j}\right)\left(S\otimes R\right)^{\dagger}=(X_{j}-\lambda_{j})\otimes\Gamma_{j}.
\end{equation*}
Therefore 
\begin{equation*}
\left(S\otimes R\right)L_{\boldsymbol{\lambda}}(X_{1},\dots,X_{d})\left(S\otimes R\right)^{\dagger}=L_{\boldsymbol{\gamma}}(X_{1},\dots,X_{d})
\end{equation*}
and this tells us that the spectra of the two localizers are equal.

When $d$ is odd,  $\Gamma_{1},\Gamma_{2},\dots,\Gamma_{d-1},\Gamma_{d}$ is no longer equivelent to the original choice of matrices, but we can find a unitary  matrix $R$ so that $R\Gamma_{j}R^{\dagger}=-\Gamma_{j}$ for $j=1,\dots,d$.  In this case we find 
\begin{equation*}
\left(S\otimes R\right)L_{\boldsymbol{\lambda}}(X_{1},\dots,X_{d})\left(S\otimes R\right)^{\dagger}=-L_{\boldsymbol{\gamma}}(X_{1},\dots,X_{d})
\end{equation*}
which is good enough, as we are after the absolute values of the eigenvalues.
\end{proof}

\section{Edge and bulk states \label{sec:bulkEdge}}

Next we look at examples of the quadratic spectrum for matrices that represent observables in a physically interesting system. Although our analysis begins with a general two-dimensional lattice, we will later specialize to a specific model that consists of half of the low-energy tight-binding model for HgTe \cite{qi_topological_2006,konig_quantum_2007,konig_quantum_2008}, where the two spin sectors are decoupled.

In a general 2D lattice, there are three observables, the Hamiltonian $H$ and matrices $X$ and $Y$ representing position.  We make the specific choice for the $\Gamma$ matrices for the Clifford representation to be the Pauli spin matrices.  Now the localizer takes the form
\begin{align*}
    L_{(x,y,E)}(X,Y,H) 
    &= (X-x)\otimes \sigma_x + (Y-y)\otimes \sigma_y + (H-E)\otimes \sigma_z \\
    &= \left[\begin{array}{cc}
H-E & X-iY-(x-iy)\\
X+iY-(x+iy) & -H+E
\end{array}\right].
\end{align*}
We also can assume that $X$ and $Y$ commute, so (\ref{difference_quad_Cliff})  becomes
\begin{equation}
\left(L_{(x,y,E)}(X,Y,H)\right)^{2}
=Q_{(x,y,E)}(X,Y,H)\otimes I+
\left[\begin{array}{cc}
0 & [H,X+iY]\\
\left([H,X+iY]\right)^{\dagger} & 0
\end{array}\right]
\end{equation}
and (\ref{PS_to_PS}) becomes
\begin{equation}
\left|\left(\mu_{\boldsymbol{\lambda}}^{Q}(X,Y,H)\right)^2-\left(\mu_{\boldsymbol{\lambda}}^{C}(X,Y,H)\right)^2\right|
\leq
\left\Vert \left[H,X+iY \right]\right\Vert .
\label{difference_quad_Cliff_2D}
\end{equation}

The units for $H$ and for $X$ and $Y$ are not necessarily compatible, so we must introduce a constant $\kappa$ that represents changing units for measuring position.  Mathematically
this just means we compute joint pseudospectra of $(\kappa X, \kappa Y, H)$.  If $\kappa$ is too close to zero the pseudospectra will only really see the system's energy spectrum.   If $\kappa$ is too large the pseudospectra will only really see position information.

\begin{figure}
    \centering
    \includegraphics{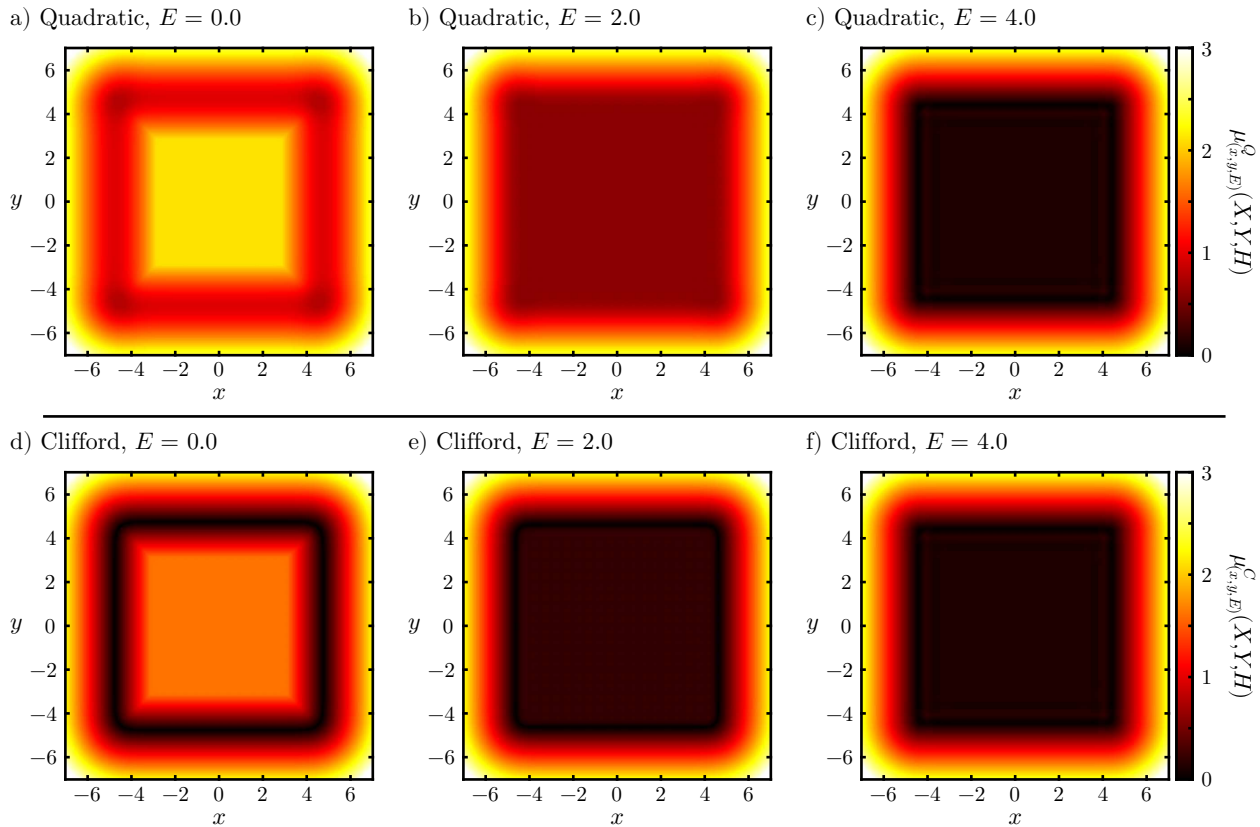}
    \caption{Here we look at slices of both flavors of pseudospectra for a very standard Chern insulator model.  This is the same model as used in \cite{LoringPseudospectra,LorHastHgTe} and is here on 20-by-20 lattice.  The lattice unit here is set to $0.5$ so the model has $x$ and $y$ coordinates in the range $-5$ to $5$.  The spectral gap in the bulk Hamiltonian is roughly between $-1$ and $1$.
    \label{fig:Chern_both_PS}}
\end{figure}

\begin{figure}
    \centering
    \includegraphics{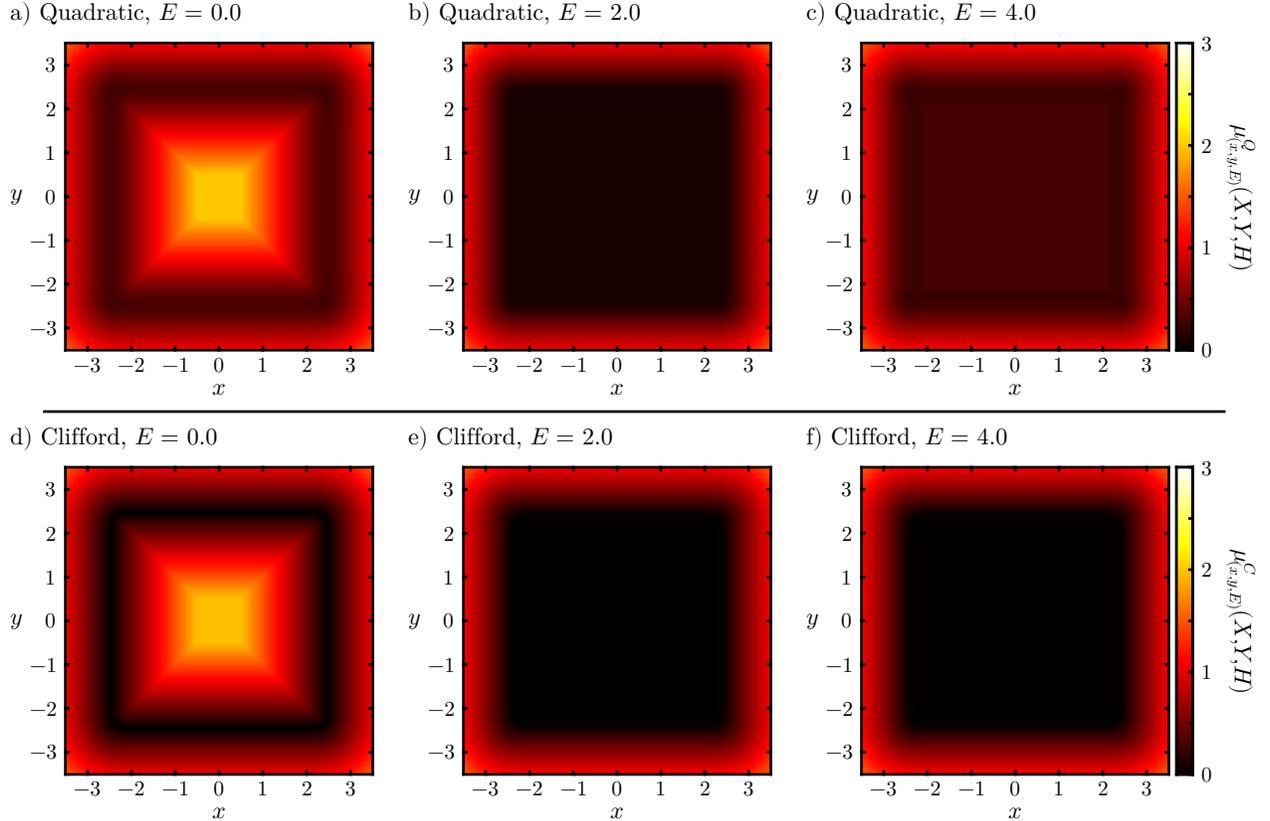}
    \caption{Here we look at slices of both flavors of pseudospectra for a very standard Chern insulator model.  This is the same model as used in \cite{LoringPseudospectra,LorHastHgTe} and is here on 100-by-100 lattice.  The lattice unit here is set to $0.05$ so the model has $x$ and $y$ coordinates in the range $-2.5$ to $2.5$.
    \label{fig:ChernBig_both_PS}}
\end{figure}

\begin{figure}
 \centering
 \includegraphics{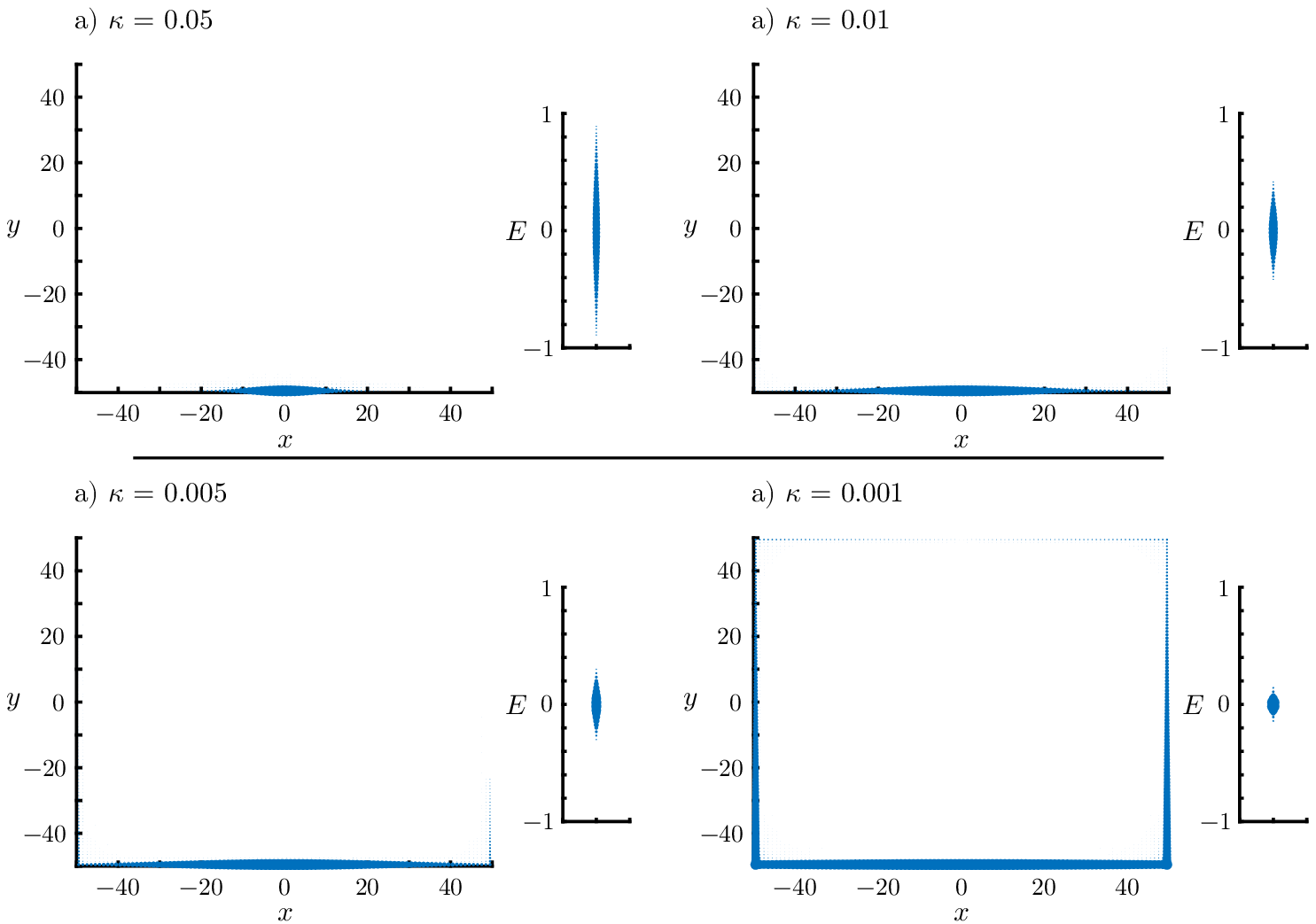}
    \caption{ Edge states for a Chern insulator found with different values of $\kappa$ in the Quadratic pseudospectra.  Each panel shows the distribution (large plot) in space and the distribution in energy (small plat) for a state selected as an eigenvalue of $Q_{\boldsymbol{\lambda}}(\kappa X, \kappa Y,H)$ for a fixed value of $\boldsymbol{\lambda} = (x,y,E)=(49,0,0)$ as $\kappa$ varies.
     \label{fig:ChernBig_localize_states}}
\end{figure}

We illustrate the quadratic and Clifford pseudospectra on a standard model of a Chern insulator \cite{konig_quantum_2008}. The real-space tight-binding model for a single spin sector of HgTe consists of a square lattice with a copy of $\mathbf{C}^2$ at each site.  That is, our Hilbert space is $\ell^2(\mathbb{Z}^2) \otimes \mathbf{C}^2$. Thus our default lattice constant is $1$ so working with $\kappa X$ and $\kappa Y$ and resets the lattice constant to $\kappa$ in units of the Hamiltonian (in this case, energy).
The on-site term in the Hamiltonian $H$ is set to
\begin{equation*}
    (C-4D)I_2 + (M-4B)\sigma_z.
\end{equation*}
There are also nearest neighbor hoping terms.  Going east these are
\begin{equation*}
    DI_2 + B\sigma_z  +A\sigma_x ,
\end{equation*}    
going north
\begin{equation*}
    DI_2 + B\sigma_z  - A\sigma_y,
\end{equation*}    
and in the remaining directions the Hermitian conjugates. Here we set $A=1$, $B=-1$, $C=0$, $D=0$ and $M =-2$ to match with the
study involving the Bott index in \cite{LorHastHgTe}, and this choice of parameters guarantees the model to have a Chern number of $-1$. The bulk spectrum is the two intervals between $\pm 1$ and $\pm 6$. 

We impose open boundary conditions and first look at a small system, just $20$-by-$20$ sites. Since $\|[H,X+iY]\| \approx 2.79$ we expect that unless we set $\kappa$ to be well less than $1$ there will be significant differences between the two pseudospectra. We also do not want $\kappa \approx 0$  since that would cause both pseudospectra to  just reflect the system's energy spectrum. In Figure \ref{fig:Chern_both_PS} we use $\kappa=0.5$.  The Clifford pseudospectrum goes to zero near the boundary when at zero energy; this is expected, as it is only when an eigenvalue of the localizer crosses zero that we can see a change in index \cite{LoringPseudospectra}.  The quadratic pseudospectrum stays relatively large, but still bounded by Eq.\ (\ref{difference_quad_Cliff_2D}).

Looking at smaller $\kappa$ will tell us more about edge states that are really localized in energy.  However, we need to look to a larger system to avoid the bulk blending into the edges.  In Figure \ref{fig:ChernBig_both_PS} we use $\kappa=0.05$ while examining a 100-by-100 lattice.  The two pseudospectra are now much closer to each other.

For every choice of $x$, $y$ and $\kappa$ that leads to small value of $\mu_{(x,y,0)}^{Q}(\kappa X,\kappa Y,H)$ there is an associated
unit vector $\boldsymbol{v}$.  This can be easily computed as an eigenvector of
$Q_{\boldsymbol{\lambda}}(\kappa X,\kappa Y,H)$ or as a right singular vector of
\begin{equation*}
   \left[\begin{array}{c}
\kappa(X-x)\\
\kappa(Y-y)\\
H-E
\end{array}\right]
\end{equation*}
as in Prop.~\ref{prop:quad_gap_four_ways}.  Figure~\ref{fig:ChernBig_localize_states} shows the nature of this
state for $(x,y,E)=(49,0,0)$, which represents a point in the middle of the bottom of the lattice, slightly in from the edge.  The system is centered at $0$ and extends in both directions from $-49.5$ to $+49.5$.  This state is computed for various values
of $\kappa$.  As expected, smaller $\kappa$ results in better localization in energy and more dispersion in position.

\section{The local nature of the quadratic pseudospectrum \label{sec:locality}}

As the quadratic pseudospectrum should be relatively easier to compute numerically than its sibling, the Clifford pseudospectrum, we want as many shortcuts and optimizations as possible for computing it.  In \cite{loring2019GuideBottLocal} it was shown that truncating a system spatially had little effect on the Clifford pseudospectrum, so long as the truncation happened well away from the probe-point $(x,y)$.  Here we establish a similar bound on the effect spacial truncation can have on the quadratic pseudospectrum.

We will assume only that the first $d$ Hermitian matrices commute with each other.  We call these $X_j$ to suggest these are position observables, but that is not important.  The last  Hermitian matrix we call $H$.  

For simplicity, we will assume always $\boldsymbol{\lambda}=\boldsymbol{0}$. We can form the Hermitian matrix
\begin{equation*}
    Z =\sqrt{X_{1}^{2}+\cdots X_{d}^{2}}
\end{equation*}
which we can think of as Euclidean distance from the origin in the spacial coordinates.  In what follows, we make the simplifying assumption that $Z$ is invertible. This just means that our model cannot have a site located at exactly the origin.  Since the quadratic pseudospectrum is Lipschitz in position we can easily work around this, if needed, to get estimates that work without this
assumption.

We deal with truncation in two steps.  First we contemplate what can happen to the quadratic pseodospectrum at $\boldsymbol{0}$ if we alter how $H$ acts on parts of the Hilbert space far away from the origin.  In particular, if we set $H$ to act as zero out there using some physical assumption of locality in the system being described.  The second step is to deal with the effect of excising that part of the Hilbert space where $H$ is now acting trivially.

The first part is in the following Theorem.  This will have applications beyond truncation, as it tells us that the quadratic pseudospectrum is generally unaffected by defects far away from the the probe location.

\begin{thm}
Assume $X_{1},\dots,X_{d}$ and $H$ are Hermitian matrices of the
same size, that the $X_{j}$ commute with each other, and that 
\begin{equation*}
Z=\sqrt{X_{1}^{2}+\cdots X_{d}^{2}}
\end{equation*}
 is invertible. If $H_{0}$ is Hermitian and 
\begin{equation}
\left\Vert Z^{-1}\left(HH_{0}+HH_{0}+H_{0}^{2}\right)Z^{-1}\right\Vert \leq C
\label{eq:assumpion_on_H0}
\end{equation}
for some constant $C$ with $C<1$, then 
\begin{equation*}
(1-C)^{\frac{1}{2}}\mu_{\boldsymbol{0}}^{Q}(X_{1},\dots,X_{d},H)\leq\mu_{\boldsymbol{0}}^{Q}(X_{1},\dots,X_{d},H+H_{0})\leq(1+C)^{\frac{1}{2}}\mu_{\boldsymbol{0}}^{Q}(X_{1},\dots,X_{d},H).
\end{equation*}
\end{thm}

\begin{proof}
With the above definition of $Z$ we obtain
\begin{equation*}
  Q_{\boldsymbol{0}}(X_{1},\dots,X_{d},H)=Z^{2}+H^{2}  
\end{equation*}
and 
\begin{equation*}
 Q_{\boldsymbol{0}}(X_{1},\dots,X_{d},H+H_{0})=Z^{2}+H^{2}+HH_{0}+H_{0}H+H_{0}^{2}.
 \end{equation*}
Multiplying (\ref{eq:assumpion_on_H0}) by $Z$ on both sides leads
to 
\begin{equation}
-CZ^{2}\leq HH_{0}+HH_{0}+H_{0}^{2}\leq CZ^{2}.\label{eq:assumption_revised}
\end{equation}
 Thus
 \begin{align*}
 Q_{\boldsymbol{\text{0}}}(X_{1},\dots,X_{d},H+H_{0})
 & \leq Z^{2}+H^{2}+C|Z|^{2} \\
 &\leq (1+C)\left(Z^{2}+H^{2}\right)\\
 &= (1+C)Q_{\boldsymbol{0}}(X_{1},\dots,X_{d},H)
\end{align*}
and 
\begin{align*}
Q_{\boldsymbol{0}}(X_{1},\dots,X_{d},H+H_{0})	&\geq Z^{2}+H^{2}-C|Z|^{2}\\
	&\geq(1-C)\left(Z^{2}+H^{2}\right) \\
	&=(1-C)Q_{\boldsymbol{0}}(X_{1},\dots,X_{d},H)
\end{align*}
so 
\begin{equation*}
	(1-C)Q_{\boldsymbol{0}}(X_{1},\dots,X_{d},H)\leq Q_{\boldsymbol{\text{0}}}(X_{1},\dots,X_{d},H+H_{0})\leq(1+C)Q_{\boldsymbol{0}}(X_{1},\dots,X_{d},H).
\end{equation*}
 Since $(1/x)$ reverses order we have 
 \begin{align*}
 (1+C)^{-1}\left(Q_{\boldsymbol{0}}(X_{1},\dots,X_{d},H)\right)^{-1}
 &\leq Q_{\boldsymbol{0}}(X_{1},\dots,X_{d},H+H_{0})^{-1} \\
 &\leq(1-C)^{-1}\left(Q_{\boldsymbol{0}}(X_{1},\dots,X_{d},H)\right)^{-1}.
\end{align*}
This means
\begin{align*}
 (1+C)^{-1}\left\Vert \left(Q_{\boldsymbol{0}}(X_{1},\dots,X_{d},H)\right)^{-1}\right\Vert &\leq \left\Vert Q_{\boldsymbol{0}}(X_{1},\dots,X_{d},H+H_{0})^{-1}\right\Vert \\
 & \leq(1-C)^{-1}\left\Vert \left(Q_{\boldsymbol{0}}(X_{1},\dots,X_{d},H)\right)^{-1}\right\Vert
\end{align*}
 and finally 
\begin{align*}
  (1-C)^{\frac{1}{2}}\left\Vert \left(Q_{\boldsymbol{0}}(X_{1},\dots,X_{d},H)\right)^{-1}\right\Vert ^{-\frac{1}{2}} &\leq \left\Vert Q_{\boldsymbol{0}}(X_{1},\dots,X_{d},H+H_{0})^{-1}\right\Vert ^{-\frac{1}{2}} \\
  &\leq(1+C)^{\frac{1}{2}}\left\Vert \left(Q_{\boldsymbol{0}}(X_{1},\dots,X_{d},H)\right)^{-1}\right\Vert ^{-\frac{1}{2}}.
\end{align*}
 
\end{proof}

\begin{thm}
Assume $X_{1},\dots,X_{d}$ are diagonal matrices and $H$ is a Hermitian
matrix, all in $\boldsymbol{M}_{n}(\mathbb{C})$, and set 
\begin{equation*}
Z=\sqrt{X_{1}^{2}+\cdots X_{d}^{2}}.
\end{equation*}
Let $\mathcal{H}_{\rho}$ denote the Hilbert subspace of $\mathbb{C}^{n}$
corresponding to standard basis vectors where $Z$ takes value at
most $\rho$. Let $X_{1}^{\rho},\dots,X_{d}^{\rho}$ and $H^{\rho}$
denote the compressions to $\mathcal{H}_{\rho}$. If $H$ acts trivially
on the complement of $\mathcal{H}_{\rho}$ then 
\begin{equation*}
Q_{\boldsymbol{0}}(X_{1},\dots,X_{d},H)=\min\left(\rho,Q_{\boldsymbol{0}}(X_{1}^{\rho},\dots,X_{d}^{\rho},H^{\rho})\right).
\end{equation*}

\end{thm}

\begin{proof}
The difference between \begin{equation*}
\mu_{\boldsymbol{0}}^{Q}(X_{1},\dots,X_{d},H)
\end{equation*} 
and 
\begin{equation*}
\mu_{\boldsymbol{0}}^{Q}(X_{1}^{\rho},\dots,X_{d}^{\rho},H^{\rho})
\end{equation*}
is the addition of many small summands of the form
\begin{equation*}
\sum_{j=1}^{d}\lambda_{j}\Gamma_{j}+0\Gamma_{d+1}.
\end{equation*}
Each contributes  two points to the spectrum, specifically
$\pm\Big(\sum_{j=1}^{d}\lambda_{j}\Big)^{\frac{1}{2}}$.
\end{proof}

\section{Conclusion \label{sec:conc}}

In conclusion, we have established the necessary definitions and theorems for both understanding the utility of the quadratic composite operator and quadratic pseudospectrum in physical systems, as well as the relationship between the quadratic and Clifford pseudospectra. Moreover, we have proven that the quadratic pseudospectrum is local, which has two important consequences. First, the numerical difficultly in calculating the quadratic pseudospectrum plateaus beyond a certain system size, where the error in the calculation incurred through the truncation becomes negligible.  
Second, this provides us a new tool to related bound states in a large system to bound states in a more easily understood truncated system.  This was already possible using the truncation bound known for the localizer in \cite{loring2019GuideBottLocal}, but the bounds in Sect.\ \ref{sec:locality} are simpler.

On a more fundamental level, the quadratic composite operator and pseudospectrum represent the most straightforward method for approaching systems with incompatible observables, as it both minimizes the eigen-error in the joint approximate spectrum and does not increase the computational complexity of the system. If the system is suspected of possessing non-trivial $K$-theory, this can then be calculated using the localizer where the quadratic gap is maximized, which will typically coincide with large localizer gaps due to Prop.\ \ref{prop:Diff_PS_bound}. Similarly, any topological boundary-localized states that exist in systems with non-trivial $K$-theory in their bulk can also be found near minima in the system's quadratic gap. In particular, this may be especially relevant for studying systems whose topology is not yet known to be connected to another pseudospectra, such as non-Hermitian systems that are known to possess non-trivial topology both theoretically \cite{lee_anomalous_2016,leykam_edge_2017,shen_topological_2018,cerjan_effects_2018,kunst_biorthogonal_2018,yao_edge_2018,gong_topological_2018,wojcik_homotopy_2020,bergholtz_exceptional_2021} and experimentally \cite{zeuner_observation_2015,weimann_topologically_2017,kremer_demonstration_2019,cerjan_experimental_2019}.

\section*{Acknowledgements}

T.L. acknowledges support from the National Science Foundation, grant DMS-2110398.
A.C. and T.L. acknowledge support from the Center for Integrated Nanotechnologies, an Office of Science User Facility operated for the U.S.\ Department of Energy (DOE) Office of Science, and the Laboratory Directed Research and Development program at Sandia National Laboratories. Sandia National Laboratories is a multimission laboratory managed and operated by National Technology \& Engineering Solutions of Sandia, LLC, a wholly owned subsidiary of Honeywell International, Inc., for the U.S.\ DOE's National Nuclear Security Administration under contract DE-NA-0003525. The views expressed in the article do not necessarily represent the views of the U.S.\ DOE or the United States Government. F. V. acknowledges support from the Scientific Computing Innovation Center of UNAH, as part of the researh project PI-063-DICIHT.

\bibliographystyle{plain}
\bibliography{sample}

\end{document}